\documentclass[preprint,12pt]{elsarticle}

\usepackage{natbib}
\usepackage{graphicx}
\usepackage{amssymb}

\usepackage{lineno}

\usepackage{amsthm}

\theoremstyle{definition}
\newtheorem{definition}{Definition}

\theoremstyle{remark}

\newtheorem{theorem}{Theorem}
\newtheorem{lemma}[theorem]{Lemma}

\begin{document}

\begin{frontmatter}


\title{Tight Lower Bound for Average Number of Terms in Optimal Double-base Number System}


\author[rvt]{Vorapong Suppakitpaisarn\fnref{fn1}}
\ead{vorapong@is.s.u-tokyo.ac.jp}
\address[rvt]{Graduate School of Information Science and Technology, The University of Tokyo}

\begin{abstract}
We show in this note that the average number of terms in the optimal double-base number system is in $\Omega(n / \log n)$. The lower bound matches the upper bound shown earlier by Dimitrov, Imbert, and Mishra (Math. of Comp. 2008).
\end{abstract}


\end{frontmatter}


\section{Introduction and Notations}
\label{S:1}

Given a non-negative integer $m$. A tuple $\left[k, X = \langle x_i \rangle_{i = 1}^{k}, Y = \langle y_i \rangle_{i = 1}^k\right]$ where $k, x_i, y_i$ are non-negative integers is a representation of $m$ in double-base number system if $\sum\limits_{i = 1}^k 2^{x_i}3^{y_i} = m$. For some $m$, there might be more than one representations of $m$ in the system. For example, both $\left[2, \langle 0,2 \rangle, \langle 1,1 \rangle\right]$ and $\left[4, \langle 0,1,2,3 \rangle, \langle 0,0,0,0 \rangle\right]$ are representations of $15$. 

For a representation $\left[k, X = \langle x_i \rangle_{i = 1}^{k}, Y = \langle y_i \rangle_{i = 1}^k\right]$, we call $k$ as the number of terms of the representation. A representation with $k$ terms is a minimum representation of $m$ if there is no representation of $m$ with $k'$ terms such as $k' < k$. The number of terms of a minimum representation is denoted by $k^*_m$ in this paper.

To speed up the calculation of scalar multiplication in elliptic curve cryptography, many researchers devise algorithms to calculate a representation with small number of terms for a given $m$. Those include the algorithm by Dimitrov, Imbert, and Mishra \cite{dimitrov2008double}. Let $k'_m$ be the number of terms of representations obtained from the algorithm. The authors have shown that $k'_m \in O(\lg m / \lg \lg m)$ for all $m$. The result is surprising as, in all variations of binary representations, there are infinite number of $m$ of which all representations have $\Omega(\log m)$ terms \cite{suppakitpaisarn2014worst}. It leads to a smaller asymptotic complexity for calculating the scalar multiplication.

One may ask if the asymptotic complexity can be even smaller with the double-base number system. However, in \cite{krenn2020minimal}, we have shown that a smaller asymptotic complexity cannot be obtained. There are infinite number of $m$ of which $k_m^* \in \Omega(\lg m / \lg \lg m)$. 

Until now, we have discuss the worst-case computation time of the scalar multiplication.
However, in literature, it is more common to analyze the average-case computation time than to analyze the worst-case computation time \cite{avanzi2002multi, dahmen2005advanced,imai2015improving}. The average-case computation time usually depends on $\mathcal{A}'(n) := \sum\limits_{m = 0}^{2^n - 1} k_m' / 2^n$ when $n$ is a positive integer and $k_m'$ is the number of terms in the representation obtained from an algorithm. When the algorithm is that proposed in \cite{dimitrov2008double}, we have 
$$\mathcal{A}'(n) = \sum\limits_{m = 0}^{2^n - 1} \frac{k_m'}{2^n} = \sum\limits_{m = 0}^{2^n - 1} O(\frac{\lg m}{\lg \lg m}) / {2^n} = \sum\limits_{m = 0}^{2^n - 1} O(n / \lg n) / {2^n} = O(n / \lg n).$$
Let $\mathcal{A}^*(n) := \sum\limits_{m = 0}^{2^n - 1} k_m^* / 2^n$. We know that $\mathcal{A}^*(n) \leq \mathcal{A}'(n) = O(n / \lg n)$.

Although it is known that the worst-case asymptotic complexity using double base number system cannot be further improved, we may be able to further improve the average-case asymptotic complexity. We have tried to improve our algorithm and our analysis to have $\mathcal{A}'(n) \in o(n / \lg n)$, but have not been successful.

\section{Our Result}

We show that the asymptotic of $\mathcal{A}'(n)$ cannot be further improved, i.e. $\mathcal{A}^*(n) \in \Omega(n / \lg n)$. This implies that the upper bound of $\mathcal{A}'(n)$ is asymptotically tight, and the algorithm in \cite{dimitrov2008double} is asymptotically optimal also on average case.

We will use the concept of prefix code and its properties to show the tightness. The prefix code can be defined as follows:

\begin{definition}[Prefix code] Let $N$ be a positive integer and, for $i \in \{0, \dots, N-1\}$, let $c_i \in \{0, 1\}^*$. We say that $c_0, \dots, c_{N-1}$ is a prefix code of $0, \dots, N-1$ if $c_i$ is not a prefix of $c_j$ for $i \neq j$. 
\end{definition}

We can define an optimal prefix code as follows:

\begin{definition}[Optimal prefix code] Let $c_0, \dots, c_{N-1}$ be a prefix code of $0,\dots, N-1$, let $p_0, \dots, p_{N-1}$ be a probability distribution on $0, \dots, N-1$, and let $|c|$ be the length of $c \in \{0, 1\}^*$

We say that $c_0, \dots, c_{N-1}$ is an optimal prefix code under the probability distribution $p_0, \dots, p_{N-1}$ if there is no prefix code $c'_0, \dots, c'_{N-1}$ of $0, \dots, N-1$ such that $\sum\limits_{i = 0}^{N - 1} |c'_i| p_i < \sum\limits_{i = 0}^{N - 1} |c_i| p_i$.
\end{definition}

We will use an algorithm in \cite{huffman1952method} to show our result. The algorithm takes a probability distribution of $0, \dots, N-1$ and gives a prefix code $c_0, \dots, c_{N-1}$ as an output. The code obtained from the algorithm is called Huffman code. The author of the paper has shown the following lemma.  

\begin{lemma} [Optimality of Huffman code \cite{huffman1952method}] 
The output code obtained from the algorithm in \cite{huffman1952method} is an optimal prefix code under the input probability distribution. 
\end{lemma}

By Lemma 1, we can obtain the following lemma.

\begin{lemma}
Suppose that $N$ can be written in the form of $2^n$ for some $n \in \mathbb{Z}_{>0}$. Then, there is no prefix code $c_0, \dots, c_{N-1}$ such that $\sum\limits_{i = 0}^{N-1} \frac{|c_i|}{N} < n$.
\end{lemma}

\begin{proof}
We omit the description of the algorithm of \cite{huffman1952method} in this paper. However, when $N = 2^n$ for some $n$, $p_i = 1/N$ for all $i \in \{0, \dots, N-1\}$, and $c_0, \dots, c_{N-1}$ is a prefix code obtained from the algorithm, we know that $|c_i| = n$ for all $i$. By the optimality of $c_1, \dots, c_{N-1}$, there is no prefix code $c_1', \dots, c_{N-1}'$ such that $\sum\limits_{i = 0}^{N - 1} |c'_i| p_i < \sum\limits_{i = 0}^{N - 1} |c_i| p_i$ or $\sum\limits_{i = 0}^{N - 1} |c'_i|/N < \sum\limits_{i = 0}^{N - 1} |c_i|/N = n$.
\end{proof}

Let a tuple $\left[k_m^*, X^* = \langle x^*_i \rangle_{i = 1}^{k_m^*}, Y^* = \langle y^*_i \rangle_{i = 1}^{k_m^*}\right]$ be a minimum representation of $m$ in double-base number system. Recall from Section 1 that $\sum\limits_{i = 1}^{k_m^*} 2^{x^*_i}3^{y^*_i} = m$. We know that, for all $i$, $x^*_i, y^*_i \leq \lg m$, otherwise the term $2^{x^*_i}3^{x^*_i}$ and the summation $\sum\limits_{i = 1}^k 2^{x^*_i}3^{y^*_i}$ would be larger than $m$. Therefore, we can represent $x^*_i$ and $y^*_i$ using a binary representation length $\lg n$ when $m < 2^n$. Also, we know that the number of terms in the minimum representation would not be larger than $n$ when $m < 2^n$. We can represent $k_m^*$ using a $(\lg n)$-bit binary representation. 

Let $b[x] \in \{0,1\}^{\lg n}$ be a $(\lg n)$-bit binary representation of a non-negative integer $x \in \{0, \dots, n - 1\}$. Note that the length of $b[x]$ is fixed to $\lg n$ and the bit string $b[x]$ begins with $0$ when $x < 2^{n - 1}$. We define the code $c''_0, \dots, c''_{2^n - 1}$ of $0, \dots, 2^n - 1$ in the way that, for all $m$, 
$$c''_m = b[k_m^*]b[x_1^*]b[y_1^*] \dots b[x_{k_m^*}^*] b[y_{k_m^*}^*].$$  
We can show the following lemma for the code.

\begin{lemma}
The code $c''_0, \dots, c''_{2^n - 1}$ of $0, \dots, 2^n - 1$ defined above is a prefix code.
\end{lemma}

\begin{proof}
Assume a contradictory statement that there are $i \neq j$ such that $c_i''$ is a prefix of $c_j''$. We must have $b[k_i^*] = b[k_j^*]$ and $k_i^* = k_j^*$. We can then conclude that the code $c_i''$ and $c_j''$ have the same length. As $c_i''$ is a prefix of $c_j''$, we have $c_i'' = c_j''$. Thus, for $1 \leq p \leq k_m^*$, we have same $x_p^*$ and $y_p^*$ in the representations of $i$ and $j$. The integer $i$ and $j$ have the same representation in double-base number system. That is not possible when $i \neq j$.   
\end{proof}

We are now ready to show our main result.

\begin{theorem}
The average number of terms in the double-base number system, denoted by $\mathcal{A}^*(n)$, is in $\Omega(n / \lg n)$.
\end{theorem}

\begin{proof}
By the construction of $c_m''$, we have $|c_m''| = (2 k_m^* + 1) \cdot \lg n$. Then, $$\sum\limits_{i = 0}^{2^n - 1} |c_i''| / 2^n = \sum\limits_{i = 0}^{2^n - 1} (2 k_m^* + 1) \cdot \lg n / 2^n = (2 \mathcal{A}^*(n)  + 1) \cdot \lg n.$$
If $\mathcal{A}^*(n) \in o(n / \lg n)$, then 
$$\sum\limits_{i = 0}^{2^n - 1} |c_i''| / 2^n = (2\cdot o(n / \lg n) + 1) \cdot \lg n = o(n).$$
There is $n$ such that $\sum\limits_{i = 0}^{2^n - 1} |c_i''| / 2^n < n$. This contradicts our result in Lemma~2.
\end{proof}

\section{Concluding Remarks}

In this paper, we have shown a tight lower bound for the double-base number system in the previous section. Indeed, we can use the same argument to show the tight lower bound for the multi-base number system. In other words, for any constant $q$ and for any $b_1, \dots, b_q$ such that $b_i$ and $b_j$ are co-prime for $i \neq j$, the average number of terms in the summation $\sum\limits_{i = 1}^k b_1^{\beta_1}, \dots b_q^{\beta_q} = m$ would not be smaller than $\Omega(n / \lg n)$ for $m \in \{0, \dots, 2^n - 1\}$. 

The result also holds for the case that digit set is not $\{0,1\}$ as far as the exponent is in $O(\log n)$. In other words, let the digit set be $D_S$, the average number of terms in the summation $\sum\limits_{i = 1}^k d_i b_1^{\beta_1}, \dots b_q^{\beta_q} = m$ for $d_i \in D_S$ is not asymptotically smaller than that without $d_i$ unless we allow $\beta_i$ to be in $\Omega(\log n)$. This implies that our algorithm for the multi-base number system in \cite{krenn2020minimal} is tight also on the average case.

\section*{Acknowledgement}

The author would like to thank Prof. Kazuo Iwama (Kyoto University) for suggesting us to use results from information theory in this research area.

\end{document}